\documentclass[10pt,journal,compsoc]{IEEEtran}

\ifCLASSOPTIONcompsoc
  \usepackage[nocompress]{cite}
\else
  \usepackage{cite}
\fi

\usepackage{tcolorbox,enumitem,float,amsthm,amsmath,hyperref,algorithm2e, enumitem}
\usepackage{graphicx}
\usepackage{subfigure}
\usepackage{balance}
\newtheorem{theorem}{Theorem}
\newcommand{\tabincell}[2]{\begin{tabular}{@{}#1@{}}#2\end{tabular}}

\begin{document}
	
\title{SecGrid: A Secure and Efficient SGX-enabled Smart Grid System with Rich Functionalities}

\author{Shaohua~Li,
        Kaiping~Xue,~\IEEEmembership{Senior~Member,~IEEE}
\IEEEcompsocitemizethanks{\IEEEcompsocthanksitem S. Li and K. Xue are with the department of Electronic Engineering and Information Science, University of Science and Technology of China, Hefei, Anhui,
	China, 230027 (Email: kpxue@ustc.edu.cn (K. Xue)).}
\thanks{}
}

\IEEEtitleabstractindextext{%
\begin{abstract}
Smart grid adopts two-way communication and rich functionalities to gain a positive impact on the sustainability and efficiency of power usage, but on the other hand, also poses serious challenges to customers' privacy. Existing solutions in smart grid usually use cryptographic tools, such as homomorphic encryption, to protect individual privacy, which, however, can only support  limited and simple functionalities. Moreover, the resource-constrained smart meters need to perform heavy asymmetric cryptography in these solutions, which is not applied to smart grid. In this paper, we present a practical and secure SGX-enabled smart grid system, named SecGrid. Our system leverage trusted hardware SGX to ensure that grid utilities can efficiently execute rich functionalities on customers' private data, while guaranteeing their privacy. With the designed security protocols, the SecGrid only require the smart meters to perform AES encryption. Security analysis shows that SecGrid can thwart various attacks from malicious adversaries. Experimental results show that SecGrid is much faster than the existing privacy-preserving schemes in smart grid.
\end{abstract}

\begin{IEEEkeywords}
Smart Grid, Intel SGX, Data Aggregation, Dynamic Pricing, Load Forecasting, Security, Privacy.
\end{IEEEkeywords}}

\maketitle

\IEEEdisplaynontitleabstractindextext

\IEEEpeerreviewmaketitle

\IEEEraisesectionheading{\section{Introduction}\label{par:s1-introduction}}
\IEEEPARstart{S}{mart} grid integrates various information and communication technologies to achieve efficient and reliable power generation, transmission, distribution, and control \cite{erol2015energy,Review-2016-KMSKK,he2017win}. Each house will be equipped with a smart meter, which collects customers' interval data (typically minute-level or second-level power usage profile) for billing or analyzing purpose. On the one hand, these fine-grained data are used to enable real time analysis, such as dynamic pricing \cite{LLLLS-UDP-2014,Optimal-2010-RHG} and load forecasting\cite{ANALYSIS-1989-IS,raza2015review}. On the other hand, this information raises privacy concerns because it reveals important personal information and can lead to various cyberattacks \cite{mishra2015rate,liu2016leveraging}. For example, attackers can derive the appliance usage patterns of the householders from fine-grained energy usage profile \cite{Achieve-2014-zhao}.

To prevent customers' fine-grained data from disclosure, secure data aggregation schemes  \cite{LXYH-PPMA-2017,rahman2017secure,chim2015prga} have been proposed to aggregate overall power usage data. In these schemes, each smart meter encrypts data using homomorphic cryptography, such as Paillier \cite{paillier1999public} and BGN \cite{boneh2005evaluating}, then reports the ciphertext to a gateway. The gateway will compute the aggregation result on ciphertext and then report the result to control center for further analysis. Data aggregation guarantees that only overall power usage data will be known by others, thereby protecting customers' privacy. However, some important tools, such as dynamic pricing and load forecasting, which can be used to ensure grid system's stability and reliability, require the grid utilities to compute on customers' fine-grained data. Dynamic pricing is used to charge customers with dynamic prices based on their real time usage. To realize it in a privacy-preserving way, the existing schemes, like the one in \cite{LLLLS-UDP-2014}, use homomorphic encryption and sophisticated design. As for load forecasting, so far, no one has designed a practical privacy-preserving scheme for it due to the need of complex operations.

Although utilizing homomorphic encryption can realize data aggregation and dynamic pricing in a privacy-preserving way, it brings a heavy computation overhead to the grid utilities, especially for resource-constrained smart meters \cite{victoria2014guidelines,he2017cyber}. Furthermore, many useful tools, like load forecasting, are less likely to be implemented efficiently with privacy protection in the same way. In addition, if we want to realize multiple tools in one system, the computation overhead will be even higher. We refer to these tools (i.e., data aggregation, dynamic pricing, load forecasting, etc.) as functionalities.

In general, it is a trade-off between rich functionalities and strong privacy protection, as it is very hard to achieve both of the features simultaneously. However, in this paper, we accomplish both by our novel design. We present SecGrid, a secure and efficient smart grid system that possesses the properties of privacy preservation and rich functionalities. Our security model considers malicious adversaries who may control the software and even the OS of the whole grid utilities (including gateways and control center) except for the certified physical processors involved in the computation. In SecGrid, the resource-constrained smart meters only need to perform AES encryption, and the gateways can perform rich functionalities with high efficiency in a privacy-preserving way. In fact, our system can be treated as a framework, as any functionality that can be implemented obliviously is compatible with SecGrid.

Our main contribution is the design, implementation, and evaluation of this practical smart grid system. We use SGX processor, which is Intel's trusted hardware capability \cite{costan2016intel}, as a building block. Indeed, SGX does not guarantee to secure everything, and we need to cope with many challenges not addressed by the hardware. The first is to establish a secret key between a smart meter and a gateway. Since smart meter is resource-constrained, it cannot perform heavy cryptographic schemes, such as Diffie-Hellman key exchange and asymmetric encryption \cite{tan2017survey,ReadMeter-2017-SAW}. To solve this problem, a user device is introduced to participate in the initialization phase of a smart meter, which can only interact once with smart meter to complete the key exchange.

The second challenge is to guarantee data integrity for the smart meters' reports. Since SGX has no non-volatile storage, the customers' reports that need to be stored in the storage of gateway, may be tampered, removed or rolled back by a malicious software or compromised OS \cite{sangho2017inferring,T-SGX-2017-SLKP}. The existing solutions guaranteeing integrity in such case either bring heavy overhead or are not suitable for smart grid architecture \cite{ZeroTrace-2017-SGF,ROTE-2017-MAKDSGJC}. We thus propose a lightweight integrity guaranteed method for SecGrid. This method is inspired by \textit{count increment} technique of literature \cite{Ariadne-2016-RF}. During the processing of gateway, every report will be encrypted together with unique \textit{count} and \textit{nonce}, and the monotonicity and freshness of which will be verified in our proposed periodic report protocol.

The next challenge is to protect the data inside the isolated memory regions from attacks due to unsafe memory accesses. SGX provides the isolated memory regions for the programs, and thus unsafe implementation of programs can easily leak data or suffer from other attacks. By ``unsafe'' here means that the implemented codes may have memory access patterns or control flows that depend on the values of sensitive data. We thus provide the safe implementations of three functionalities, namely, data aggregation, dynamic pricing, and load forecasting, to show how the functionalities can be securely supported in SecGrid. Other challenges, such as time synchronization, gateway restart protection, etc., are also solved in our system.
In summary we make the following contributions:

\begin{itemize}[leftmargin=*]
	\item We present SecGrid, a practical smart grid system supporting rich functionalities while guaranteeing customers' privacy. In our design, the smart meters only need to perform AES encryption instead of heavy cryptography.
	
	\item Our system is compatible with any functions that can be implemented obliviously in smart grid. To better present our system, we implement three commonly used functions, namely data aggregation, dynamic pricing, and load forecasting, with strong data obliviousness.
	
	\item The security analysis indicates that our design is secure against malicious adversaries. Also, the experimental results show that the proposed protocols can be completed efficiently, and the runtime of three functions has around $10^3\times$ improvement compared with existing solutions.
\end{itemize}

The rest of the paper is organized as follows: Section \ref{par:s2-relatedwork} enumerates the related works of Intel SGX and rich functionalities in smart grid. Then we present some preliminaries of our system in Section \ref{par:s4-buildingblocks}. In Section \ref{par:s3-systemsecuritymodel}, we introduce our system model and security model. We illustrate protocols in detail about initialization, periodic report and gateways restart in Section \ref{par:s6-register}, followed by implementation of functions in Section \ref{par:s7-functions}. In Section \ref{par:s11-security} and \ref{par:s12-performance}, we analyze the security of our design and evaluate the performance respectively. Finally, Section \ref{par:s13-conclusion} makes a conclusion.

\section{Related Work}\label{par:s2-relatedwork}
\subsection{Functionalities for Smart Grid}
To enable rich functionalities in smart grid, customers' consumption data need to be collected for analysis. However, customers' privacy may be leaked out unconsciously during the procedures.
To guarantee privacy and also enable functions carried out in smart grid at the same time, many cryptography-driven schemes have been proposed. One popular privacy-preserving mechanism is secure data aggregation \cite{LXYH-PPMA-2017,rahman2017secure,chim2015prga}, which aggregates customers' consumption data of a specific region through homomorphic encryption. Another widely researched topic is dynamic pricing \cite{tushar2017price,ye2016real,misra2015d2p}. However, privacy protection schemes, like the one in \cite{LLLLS-UDP-2014}, can only handle simple pricing models with cryptography tools. Some real time pricing models \cite{Optimal-2010-RHG} are not likely to realize efficiently in a privacy preserving way. Other famous functions, such as load forecasting \cite{ahmad2017accurate,chan2012load,ANALYSIS-1989-IS}, that are very useful to improve the grid's performance, also suffer the same problem.


\subsection{Intel SGX}\label{par:s2a-securesgx}
Intel SGX provides \textit{isolated} execution spaces, named enclaves. Programs in enclaves can process data in plaintext. But any software or even the OS on the same platform, cannot observe the data content inside enclaves \cite{MPML-2016-OSFMNVC,VC3-2015-SCFGPMR}. 
Ohrimenko et al. \cite{MPML-2016-OSFMNVC} showed how to outsource model training to untrusted servers. To make memory accesses data-independent, which is not protected by SGX, their system uses padding and other tricks to hide access pattern. 
VC3\cite{VC3-2015-SCFGPMR} implemented MapReduce in distributed servers with confidentiality and verifiability.  Opaque \cite{Opaque-2017-ZDBPGS} is an encrypted data analytics platform over Spark SQL, which uses oblivious sorting for data processing in an encrypted database. 
Iron \cite{Iron-2016-FVBG} implements some interesting but heavy primitives in cryptography. The SGX implementations are efficient and practical. The access pattern, not protected by SGX, is hidden by oblivious comparison functions. 
Town Crier \cite{TownCrier-2016-ZCCJS} is a work that provides authenticated data feed from external trusted sources for smart contracts.
ZeroTrace \cite{ZeroTrace-2017-SGF} provides an oblivious data storage system from SGX to access external storage, which minimizes the response time.


\section{Preliminaries}\label{par:s4-buildingblocks}
\subsection{Enclaves in SGX}\label{par:s3a-hardware}
SGX refers to Intel Software Guard Extension, a set of CPU extensions, which can provide isolated execution environments, named \textbf{enclaves}, to protect the confidentiality and integrity of the data against all other software, even a compromised OS, on the platform. When a platform is equipped with a SGX-enabled CPU (such as the gateways and control center in our system), in addition to the enclave, the memory, BIOS, I/O and even power are treated as potentially untrusted. The general processing flow in enclave is shown in Fig. \ref{fig:enclave}. The encrypted data firstly will be transmitted into enclave for decryption. Then the decrypted data will be the input of some function $f$. Finally, the output of $f$ will be encrypted and then sent to the outside of the enclave.

\begin{figure}[!h]
	\centering
	\includegraphics[width=\linewidth]{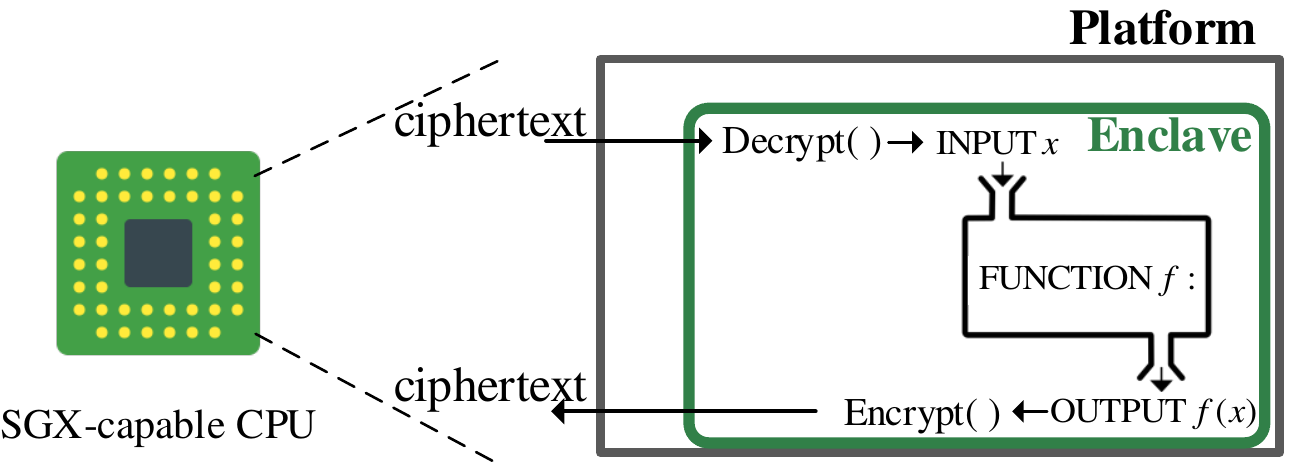}
	\caption{General processing flow in enclave\label{fig:enclave}}
	\vspace{-1px}
\end{figure}

SGX provides two core operations, \textit{sealing} and \textit{remote attestation}, which will be used in our system. \textit{Sealing} is for storing data securely outside of the enclave. \textit{Remote attestation} is for a remote party to verify the legitimacy of the enclave (i.e., the enclave is created by a legal SGX-capable CPU and the code is correctly loaded in the enclave). The details are as follows:
\begin{itemize}
\item \textit{Sealing.} Each SGX-capable CPU has a hardware-protected sealed key called \textit{Root Seal Key} that cannot be stolen or forged. An enclave can derive a \textit{Seal Key} from the \textit{Root Seal Key} using instruction {\small{\texttt{EGETKEY}}}. This key is specific to the enclave, and other enclaves cannot derive the same key. But the same enclave can always get this key even if it is destroyed and restarted. \textit{Seal Key} is used to encrypt and authenticate data stored outside of the enclave.

\item \textit{Remote attestation.} SGX allows a remote party to check whether the code is correctly loaded in an enclave. When an enclave is created, the CPU will generate a hash of the state of the loaded code and static data, known as \textit{measurement}, and a \textit{report} that contains the \textit{measurement} and optional self-defined data (e.g. a new generated public key). Then the software that created the enclave can ask for a \textit{quote}, which consists of a \textit{report} and its signature signed with a hardware-protected attestation key. Remote parties can verify the \textit{quote} by contacting the Intel Attestation Server. Such procedure is known as \textbf{\textit{remote attestation}}. The detailed operations can be found in \cite{costan2016intel}.
\end{itemize}

\subsection{Attacks against Enclave}
The protection of SGX is restricted in CPU. Although data are encrypted, the memory access patterns may leak the privacy of data inside enclave \cite{ZeroTrace-2017-SGF}. Branches in program like \texttt{if-else} make data-dependent running patterns, and enclaves that are running such programs suffer from cache-timing attack \cite{VC3-2015-SCFGPMR,MPML-2016-OSFMNVC}. Access to storage outside enclave exposes address to a PCI-e bus listener or an operating system (OS), who can create page faults. Furthermore, if the address is data-dependent, there will be page-fault attacks \cite{ZeroTrace-2017-SGF,T-SGX-2017-SLKP}.

Also, rollback attack breaks the freshness of external data \cite{ROTE-2017-MAKDSGJC}. Some data that requires a long-term preservation should reside in persistent storage, like disk. These data can be rolled back to a previous version by a compromised OS. For example, replacing the reported power usage data at time slot $t$ with $t-1$. In preventing rollback attack, 
integrity guarantee from Merkle tree \cite{ZeroTrace-2017-SGF}  fails when the platform is restarted, while other methods are either too heavy or unable to be applied to smart grid architecture \cite{ROTE-2017-MAKDSGJC,VC3-2015-SCFGPMR}. Therefore, we need to develop new approach in our design to prevent such attack.


\section{System Model and Security Model}\label{par:s3-systemsecuritymodel}
\subsection{System Model}\label{par:s3a-systemmodel}
\begin{figure}[t]
	\centering
	\includegraphics[width=\linewidth]{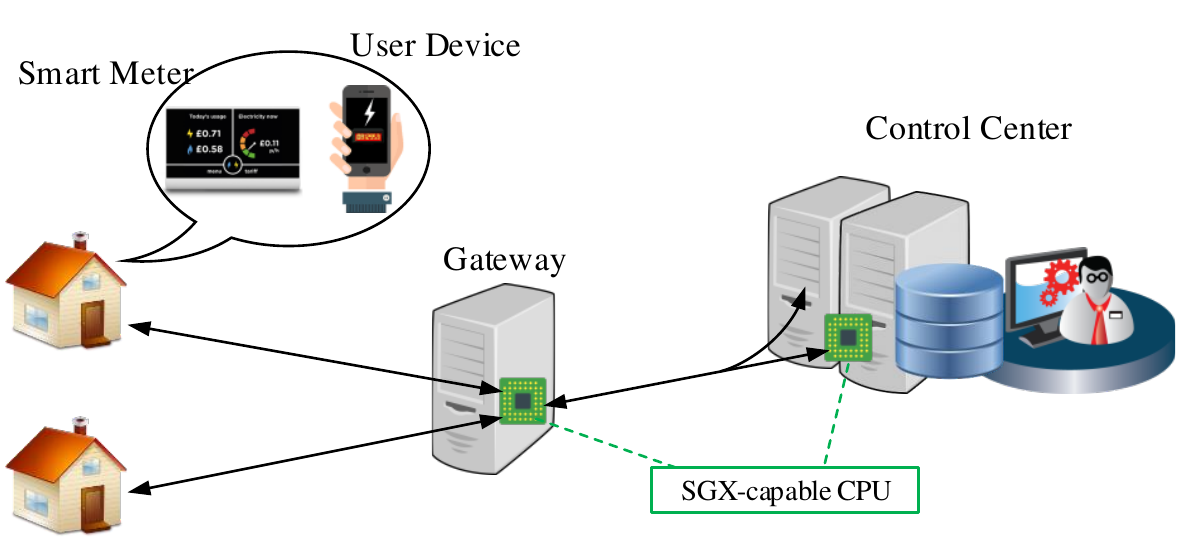}
	\caption{System Architecture}
	\label{figure:archi}
\end{figure}

Our system adopts the typical architecture of smart grid, which is shown in Fig. \ref{figure:archi}. It contains a control center, gateways in residential area, smart meters and user devices in home area.
Control center and gateways are able to create enclaves, called the \textit{control enclave} and the \textit{gateway enclave}, respectively.

\smallskip
\noindent\textbf{Control Center (CC).} 
CC collects data and responds to requests from/to each gateway. It has a enclave, called \textit{control enclave}. CC also has all the initialization keys of smart meters ($\mathcal{K}_\mathsf{init}^\mathtt{i}$ for each smart meter $i$). These keys are used in the initialization phase of smart meters, and can be accessed by the \textit{control enclave}. To prevent the CC or other potential attackers from knowing the smart meter's key, the \textit{control enclave} takes Merkle tree based techniques to access these keys.

\smallskip
\noindent\textbf{Gateways (GW).} A GW runs a secure enclave, called \textit{gateway enclave}, which directly collects and processes the data reported by smart meters. Many functionalities, such as data aggregation, dynamic pricing, and load forecasting, can be performed inside the \textit{gateway enclave}. The \textit{gateway enclave} can establish shared keys with smart meters with the help of user device.

\smallskip
\noindent\textbf{Smart Meters (SM).} Every house is equipped with a SM to collect the power usage data and report the encrypted data to \textit{gateway enclave} periodically, e.g., every 15 minutes. Considering the constrained resource of smart meter, the only cipher used here is AES-GCM used to guarantee both confidentiality and integrity of data. Each SM contains an initialization key $\mathcal{K}_\mathsf{init}^\mathtt{i}$ that is used to establish a new secret key $\mathcal{K}_\mathtt{i}$ between the SM and the \textit{gateway enclave}.

\smallskip
\noindent\textbf{User Devices (UD).} A user needs a device to help his/her SM establish a secret key with \textit{gateway enclave} in the initialization phase. The device can be a smartphone that is installed with an official application so that it can participate in the initialization of a new SM. UD has sufficient computing capability to run asymmetric cryptography algorithms and verify remote attestation. 

\subsection{Security Model}
We assume that a malicious adversary who can control all the software, including the OS, in the CC and GW, tries to violate the \textit{confidentiality} and \textit{integrity} of customers' private data by performing the following attacks. The adversary can read, block, modify, and replay all messages sent by/to a secure enclave. The adversary is also able to observe memory access pattern and infer control flow in an enclave process, i.e., launching side-channel attacks. In particular, the adversary may perform rollback attack, that is, to replace the sealed data with a previous version.

We assume that the adversary cannot compromise the secure enclaves and the relevant enclave keys (e.g. $Seal Key$ and attestation key).
The adversary cannot break cryptographic primitives used in our system, i.e AES-GCM, Diffie-Hellman key exchange, etc. Compromising the user device, denial-of-service and physical attacks, such as power analysis, are out of the scope of this paper. 

The adversary is assumed to not want to trigger alarm. Although the adversary can perform various attacks, he/she does not want to be detected by the smart grid system. 

\section{System Design}\label{par:s6-register}

\subsection{Overview}
Our goal is to guarantee customers' privacy while enabling rich functionalities in smart grid. Rich functionalities refer to various demand side management functionalities that process customers' private data to improve the grid's performance. These rich functionalities are hard to realized efficiently by cryptography-based schemes.

In SecGrid, a GW runs rich functionalities inside the \textit{gateway enclave}. To prevent side-channel attacks, these functionalities should not contain data-dependent operations, that is, they need to be implemented obliviously. We discuss the possible leakage that may be introduced by three popular functionalities, namely, data aggregation, dynamic pricing, and load forecasting, and provide the secure implementations, which will be described in Section \ref{par:s7-functions}.

To guarantee the confidentiality and integrity of input and output data of rich functionalities, in SecGrid, we develop a new periodic report protocol to secure the transmission and storage of customers' private data. Each SM's reported data is encrypted using AES-GCM, and contains two new parameters, $\mathtt{nonce}$ and $\mathtt{ctr}$. These parameters are carefully used to resist various attacks, such as replay attack and rollback attack, which, however, require more complex measures to prevent in other solutions. The secret key used in this protocol is established by our proposed initialization protocols for SMs and GWs. Considering the robustness of our system and preventing the adversaries from restarting the \textit{gateway enclave}, we propose a status restoring protocol for GW, which can avoid data loss due to the restart of GW or \textit{gateway enclave}.

\subsection{CC/GW Initialization}\label{par:systeminit}

This initialization phase involves the \textit{control enclave} and \textit{gateway enclave}. The \textit{control enclave} will be initialized at the system setup, and when the \textit{gateway enclave} starts to work, it will interact with the \textit{control enclave} to authenticate each other as well as share symmetric/asymmetric keys and time information.

The initialization protocol is shown in Fig. \ref{fig:system_init}. These six steps can be divided into two stages: \textbf{Attest and Key Exchange} and \textbf{Time Sync}. During the first stage, the \textit{gateway enclave} and the \textit{control enclave} attest each other through \textit{remote attestation}, and then use Diffie-Hellman key exchange to share a secret key. In the second stage, the \textit{control enclave} synchronizes its time to the \textit{gateway enclave} as well as confirms the shared keys. The details are shown as follows.

\smallskip
\noindent\textbf{Stage 1 [Attest and Key Exchange]}
\begin{enumerate}[label=\large{\textcircled{\small{\arabic*}}}, leftmargin=*]
	\item The \textit{gateway enclave} generates a public/private key pair ($\mathcal{PK}_\mathsf{gw}$, $\mathcal{SK}_\mathsf{gw}$) for a CCA2-secure public key cryptosystem.
	
	\item The \textit{gateway enclave} sends its \textit{remote attestation} message to the \textit{control enclave}, which contains $\mathcal{PK}_\mathsf{gw}$ and a generated Diffie-Hellman parameter $g^a$.
	
	\item Upon receiving the \textit{remote attestation}, the \textit{control enclave} verifies its legitimacy (the detailed verification phase of \textit{remote attestation} is described in Section \ref{par:s3a-hardware}). Then, it generates another Diffie-Hellman parameter $g^b$, and sends back its \textit{remote attestation} message containing $\mathcal{PK}_\mathsf{cc}$, $g^b$ and the current wall-clock time.
\end{enumerate}

\smallskip
\noindent\textbf{Stage 2 [Time Sync. and Confirm ]}
\begin{enumerate}[label=\large{\textcircled{\small{\arabic*}}}, start=4, leftmargin=*]
	\item Once the message is received, the \textit{gateway enclave} records the received time as reference time and starts the time counter from 0. After verifying the \textit{remote attestation}, the \textit{gateway enclave} obtains the time: the reference time plus the value of the time counter. Then, it encrypts the time with $g^{ab}$ and signs the ciphertext with $\mathcal{SK}_\mathsf{gw}$. The \textit{gateway enclave} sends the ciphertext and the signature to the \textit{control enclave}.
	
	\item The \textit{control enclave} verifies the signature with $\mathcal{PK}_\mathsf{gw}$, decrypts the ciphertext with $g^{ab}$, and then compares the time with local time. If all the verifications succeed, the \textit{control enclave} will seal $g^{ab}$, and return an \textit{ack} message, which is encrypted using $g^{ab}$ and signed using $\mathcal{SK}_\mathsf{cc}$.
	
	\item Upon receiving the \textit{ack}, the \textit{gateway enclave} verifies the signature using $\mathcal{PK}_\mathsf{cc}$ and decrypts the ciphertext using $g^{ab}$. Then, it seals $g^{ab}$ and $\mathcal{PK}_\mathsf{cc}$.
	
	\begin{figure}[ht]
		\centering
		\includegraphics[width=0.5\textwidth]{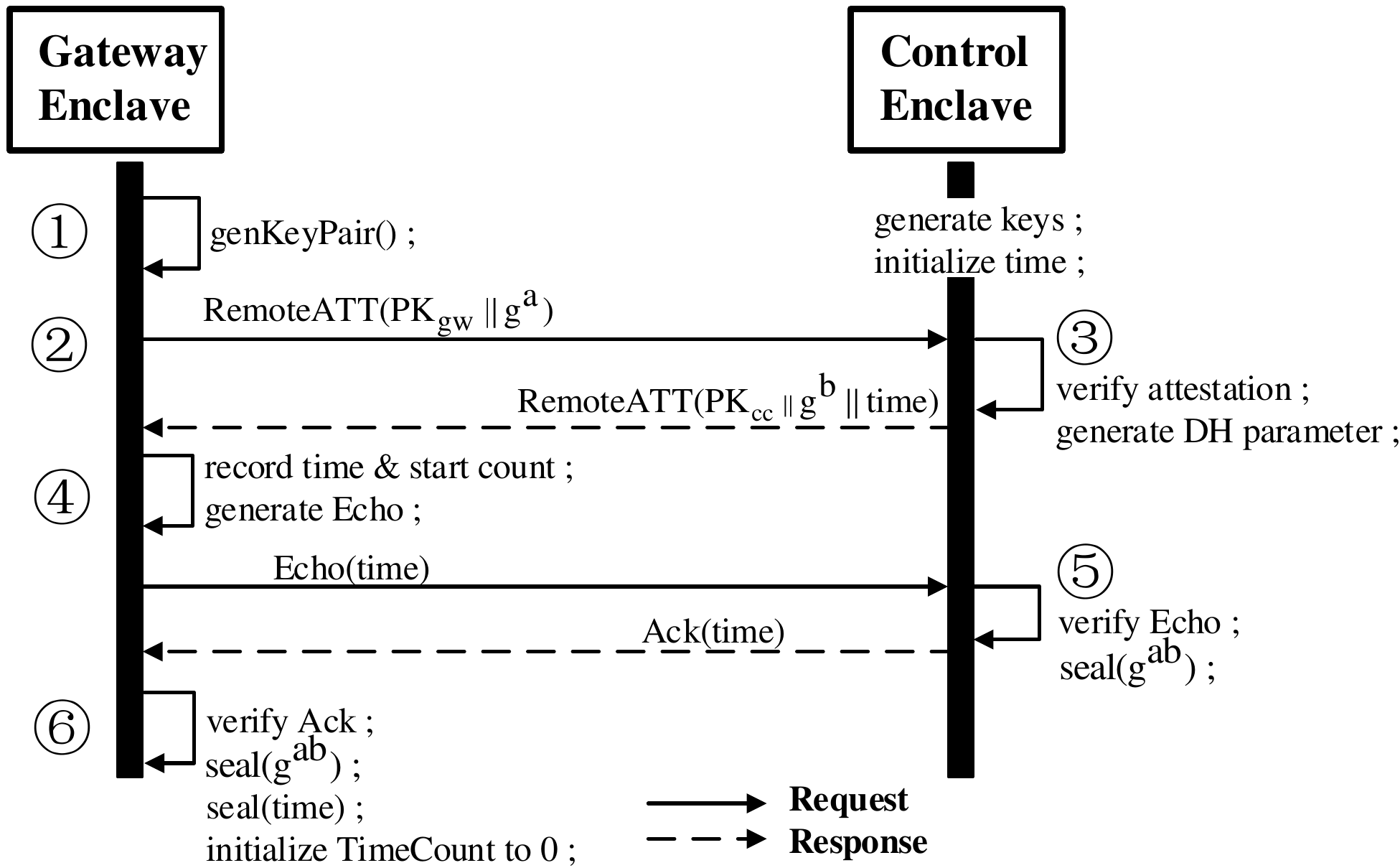}
		\caption{CC/GW initialization protocol.}
		\label{fig:system_init}
	\end{figure}
	
\end{enumerate}

\subsection{SM Initialization Protocol}\label{SM_init}
The first time when a customer accesses the smart grid, he/she first installs an official applications on UD to initialize his/her SM.
The SM initialization protocol is as shown in Fig. \ref{fig:SM_init}. This protocol can be triggered by the UD or the customer manually. We take the initialization of the \textit{i}-th SM as an example in the following description. 
The initialization phase can be divided into three stages: \textbf{Attestation}, \textbf{Key Establishment}, and \textbf{Confirm}. The UD will verify the legitimacy of the \textit{gateway enclave} according to its \textit{remote attestation}. The second stage is used to establish a secret key between the SM and the \textit{gateway enclave}, where the UD is used as a bridge. The SM has a sealed initialization key $\mathcal{K}_\mathsf{init}^\mathtt{i}$, which is also known by the \textit{control enclave}. To enable the secure use of $\mathcal{K}_\mathsf{init}^\mathtt{i}$, we utilize Merkle tree based method, proposed in ZeroTrace \cite{ZeroTrace-2017-SGF}, to guarantee integrity and freshness. Note that, after the following procedures, a fresh $\mathtt{nonce}$ will be securely obtained by the SM, and it will be used as one of the parameters in the first report of the SM.

One significant problem is how the SM sends initialization message to the UD. We can utilize existing solutions such as ZigBee based communication protocol (used in OG\&E company \cite{ReadMeter-2017-SAW}) to achieve this purpose.

\begin{figure}[t]
	\centering
	\includegraphics[width=0.45\textwidth]{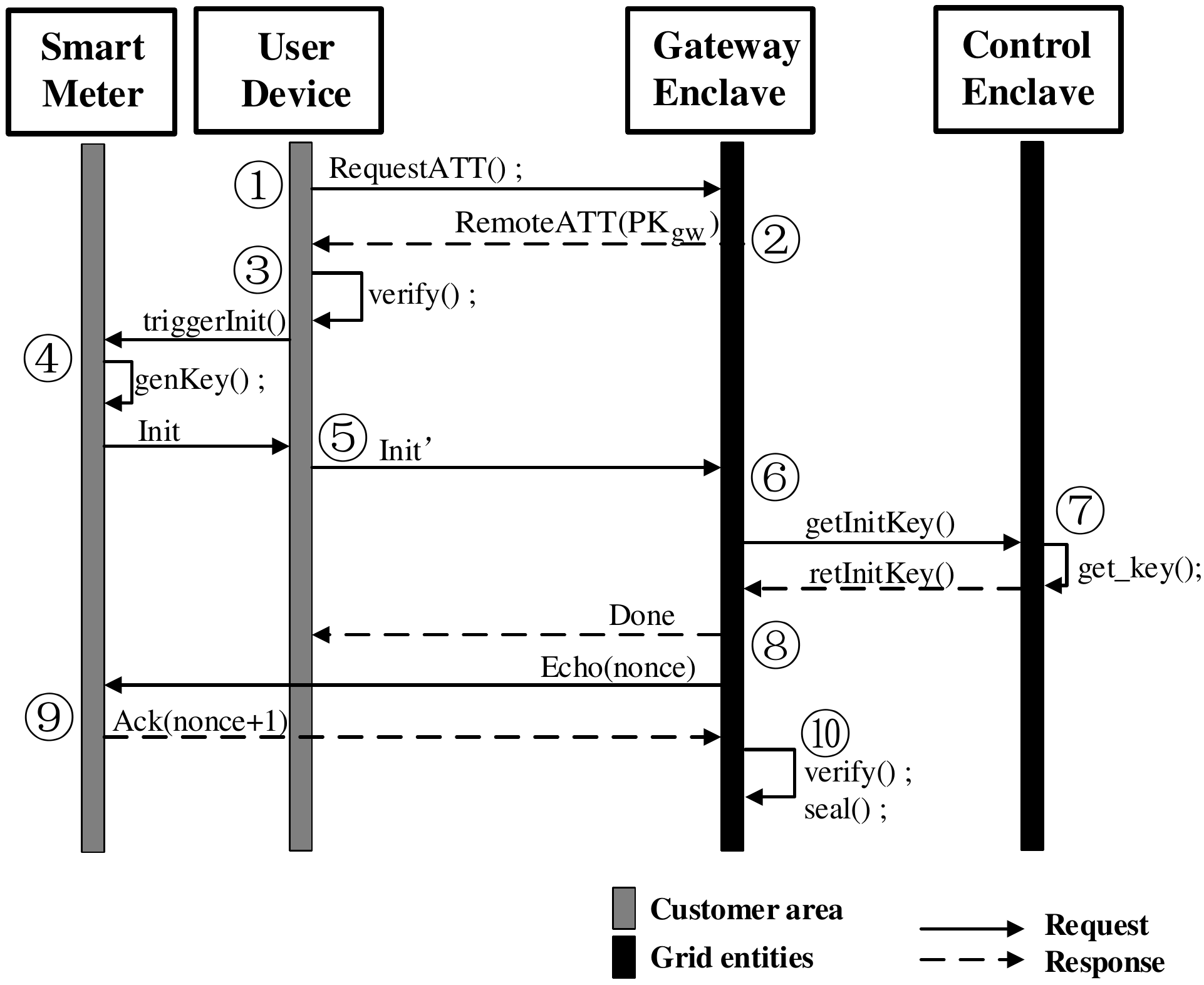}
	\caption{The SM initialization protocol.}
	\label{fig:SM_init}
\end{figure}

\smallskip
\noindent\textbf{Stage 1 [UD$\leftarrow$GE, Attestation]}
\begin{enumerate}[label=\large{\textcircled{\small{\arabic*}}}, leftmargin=*]
	\item UD starts and requests a remote attestation from the \textit{gateway enclave}.
	
	\item The \textit{gateway enclave} returns its \textit{remote attestation} message, which contains $\mathcal{PK}_\mathsf{gw}$, to the UD.
	
	\item Upon receiving the \textit{remote attestation}, the UD verifies it, and then triggers the initialization phase of the SM.
\end{enumerate}

\smallskip
\noindent\textbf{Stage 2 [SM$\leftrightarrow$GE, Key Establishment]}
\begin{enumerate}[label=\large{\textcircled{\small{\arabic*}}}, start=4, leftmargin=*]
	\item The SM starts the initialization phase. It first generates a new random key $\mathcal{K}_\mathtt{i}$, and then encrypts this key and its identifier $\mathtt{ID_i}$ with sealed initialization key $\mathcal{K}_\mathsf{init}^\mathtt{i}$. After that, the SM sends the $\mathtt{Init}$ message to the UD, where $\mathtt{Init}=\mathtt{ID_i}||E_{\mathcal{K}_\mathsf{init}^\mathtt{i}}(\mathtt{ID_i}||\mathcal{K}_\mathtt{i})$.\\
	(To be noted, the $\mathtt{Init}$ message does not need to be protected from eavesdropping since it has been encrypted. We will analyze this in detail in Section \ref{security_init}).
	
	\item The UD encrypts the $\mathtt{ID_i}$ in $\mathtt{Init}$ message using $\mathcal{PK}_\mathsf{gw}$ to generate new $\mathtt{Init'}$ message, and sends it to the \textit{gateway enclave}: $\mathtt{Init}'=\hat{E}_{\mathcal{PK}_\mathsf{gw}}(\mathtt{Init})$.
	
	\item The \textit{gateway enclave} extracts the $\mathtt{ID_i}$ with $\mathcal{SK}_\mathsf{gw}$, generates a \textit{getInitKey()} message containing the $\mathtt{ID_i}$ encrypted with $g^{ab}$, and then sends it to the \textit{control enclave}.
	
	\item Upon receiving the message, the \textit{control enclave} decrypts it and extracts the $\mathtt{ID_i}$. Then, the \textit{control enclave} obtains the corresponding $\mathcal{K}_\mathsf{init}^\mathtt{i}$, void this key and then returns it after encryption.
	
	\item The \textit{gateway enclave} can decrypt the $\mathcal{K}_\mathsf{init}^\mathtt{i}$, and obtain the $\mathcal{K}_\mathtt{i}$ from $\mathtt{Init}'$ with it. The \textit{gateway enclave} returns a \textit{Done} message to the UD to notify the initialization has succeeded, and an \textit{Echo} message to the SM, which contains an encrypted $\mathtt{nonce}$ using $\mathcal{K}_\mathtt{i}$.
\end{enumerate}

\smallskip
\noindent\textbf{Stage 3 [SM$\leftrightarrow$GE, Confirm]}
\begin{enumerate}[label=\large{\textcircled{\small{\arabic*}}}, start=9, leftmargin=*]
	\item The SM decrypts the $\mathtt{Echo}$ message, and sets the local time to \textit{time}. Then, the SM returns an $\mathtt{Ack}$ message to the \textit{gateway enclave}. This message contains the encrypted $\mathtt{nonce}+1$.
	
	\item The \textit{gateway enclave} verifies the $\mathtt{Ack}$ message, and then seals the $\mathcal{K}_\mathtt{i}$ with $\mathtt{ID_i}$ as associated data in AES-GCM.
\end{enumerate}

\subsection{Periodic Report Protocol}\label{period_report}
When the initialization phase is done, the SM shares a symmetric key with the \textit{gateway enclave}, which will be used to secure the report data. In order to ensure the data integrity and prevent replay attack, we enable an monotonic counter $\mathtt{ctr}$, which starts from 0, in SM to indicate different reports. Due to the continuity of the SM's reports, our use of $\mathtt{ctr}$ can resist rollback attack, which will be proved in our security analysis. Another parameter $\mathtt{nonce}$ is also used here to guarantee the freshness of reports. Considering that the freshness verification needs the \textit{gateway enclave} to store every $\mathtt{nonce}$, which is difficult since SGX has no permanent storage, we make clever use of cyclical nature of the reports. By letting the \textit{gateway enclave} randomly choose the $\mathtt{nonce}$ for the SM to use in the next report, we avoid the storage of each $\mathtt{nonce}$. The protocol proceeds as follows:

\begin{enumerate}[label=\large{\textcircled{\small{\arabic*}}}, start=1, leftmargin=*]
	\item When the \textit{i}-th SM needs to report, it first increases the counter $\mathtt{ctr_i}=\mathtt{ctr_i}+1$, and then generate the report $\mathtt{r_i}=\mathtt{ID_i}||E_{\mathcal{K}_\mathtt{i}}(\mathtt{ID_i}||\mathtt{m_i}||\mathtt{nonce}||\mathtt{ctr_i})$, where $\mathtt{nonce}$ is sent by the \textit{gateway enclave} during last report period. The SM reports $\mathtt{r_i}$.
	
	\item Upon receiving the report, the \textit{gateway enclave} decrypts the report with its sealed key $\mathcal{K}_\mathtt{i}$, and extracts data $\mathtt{m_i}$, $\mathtt{nonce}$, and $\mathtt{ctr_i}$. Next, the \textit{gateway enclave} verifies the correctness of $\mathtt{nonce}$ and obtains last counter $\mathtt{ctr^{old}_i}$ from storage, and then checks if $\mathtt{ctr_i}=\mathtt{ctr^{old}_i}+1$. If all passed, the \textit{gateway enclave} will seal $\mathtt{r_i}$ and process $\mathtt{m_i}$ with predefined functions. Otherwise, an error or attack may happen, the \textit{gateway enclave} will report this alarm to CC immediately. Finally, based on the outputs of functions, the \textit{gateway enclave} generates a report for CC and a response for the SM. 
	
	\item The \textit{control enclave} can process these reports in the same way as what the \textit{gateway enclave} does, and generate a response.
\end{enumerate}

\smallskip
One key challenge of above protocol is how to design and program the functions executed inside the enclaves to prevent privacy leakage. As we illustrate before, SGX is not perfectly secure. For example, software-based side channel attacks can violate the data confidentiality even if the data is inside the \textit{gateway enclave}. So we should take fully account of the secure design and implementation of these functions. We will discuss this issue in Section \ref{par:s7-functions}.

\subsection{GW Restart Protocol}\label{GW_restart}

When a \textit{gateway enclave} restarts, it needs to restore its previous state. At restart, the \textit{gateway enclave} may lose new reports that have not been sealed in the last report period. So it needs to requests these reports from the corresponding SMs. Also, it will obtain fresh time from the \textit{control enclave}. The protocol proceeds as follows:

\begin{enumerate}[label=\large{\textcircled{\small{\arabic*}}}, start=1, leftmargin=*]
	\item The \textit{gateway enclave} unseals $\mathcal{K}_\mathtt{i}$, $\mathcal{SK}_\mathsf{gw}$, $\mathcal{PK}_\mathsf{gw}$, $\mathcal{PK}_\mathsf{cc}$, $g^{ab}$ and all latest sealed reports from storage. If the unsealing procedure fails, the alarm will be triggered. Otherwise, the \textit{gateway enclave} sends request to all SMs to ask for report (containing new $\mathtt{nonce}$s for each SM) as well as to the \textit{control enclave} to ask for $\mathtt{time}$.
	\item \begin{enumerate}[label=\textbf{\alph*.}, leftmargin=*]
		\item  Each SM returns its latest report to the \textit{gateway enclave}.
		
		\item The \textit{control enclave} returns current time to the \textit{gateway enclave}.
	\end{enumerate}

	\item Upon receiving responses, the \textit{gateway enclave} verifies the time then sets the local time. Then it checks the freshness of $\mathtt{nonce}$s and if $\mathtt{ctr_i}$ in each report satisfies $\mathtt{ctr_i}=\mathtt{ctr^{old}_i}$ or $\mathtt{ctr_i}=\mathtt{ctr^{old}_i}+1$. If all passed, the \textit{gateway enclave} restores successfully. Otherwise, the restore phase may encounter a problem, and the \textit{gateway enclave} will trigger alarm to notify grid administrators.
\end{enumerate}

\section{Functions}\label{par:s7-functions}
The protocols have secured the data submission from SMs to the \textit{gateway enclave} and the \textit{control enclave}. However, the data processing inside enclaves, i.e. the execution of functions, may have data-dependent operations and is subject to side-channel attacks \cite{MPML-2016-OSFMNVC,Opaque-2017-ZDBPGS,Iron-2016-FVBG}. Programs in enclaves should be oblivious, and therefore we describe the oblivious implementation of data aggregation, dynamic pricing, and load forecasting.
In addition, we present how to support general functions in our system.

\subsection{Data Aggregation}\label{par:s7-usage-aggregate}
Secure data aggregation is used to aggregate overall power usage of all customers over a timespan as follows:
\begin{equation}
\mathsf{PowerUsage}^T_{\mathsf{area}}(t)=\sum_{i\in S}\sum_{\Delta=t}^{t+T-1}\mathsf{PowerUsage}_{i}(\Delta),
\end{equation}
where $S$ is the customer set in this area. 
Compared with cryptography-based schemes \cite{LXYH-PPMA-2017,chim2015prga} that uses Paillier or BGN cryptosystem for homomorphic computation, our system only requires symmetric encrypted data for submission and computation. To make the aggregation oblivious, the \textit{gateway enclave} adds up the power usage data with the same order of the arrival of reports and no additional leakage exists.

\subsection{Dynamic Pricing}\label{par:s9-dynamicpricing}
Typical dynamic pricing models include Time-of-Use (ToU), Critical Peak Pricing (CPP), and Real Time Pricing (RTP) \cite{Review-2016-KMSKK,Optimal-2010-RHG}. We here introduce the secure implementation of them.

\smallskip
\noindent\textbf{ToU} Electricity prices are different at peak time and at off-peak time. Peak time prices are higher than off-peak time for demand control. GW takes the price from a piecewise function:
\begin{equation}
\mathsf{PricePerUnit}(t)=
\begin{cases}
{\mathtt{p}}&\mbox{if $t\in$ off-peak time}\\
{\mathtt{p+\Delta{p}}}&\mbox{if $t\in$ {at-peak time}}
\end{cases}\label{eqn:priceperunit},
\end{equation}
where $t$ is the current time. The dynamic pricing does not have sensitive patterns, because the condition for the piecewise function (\ref{eqn:priceperunit}) is time $t$, which is open. 

But a secure and reliable time \cite{Iron-2016-FVBG,TownCrier-2016-ZCCJS} needs extra efforts, since timestamp in gateway BIOS can be tampered. As mentioned in Section \ref{par:systeminit}, the \textit{gateway enclaves} can obtain the time from the \textit{control enclave}, i.e., gets the time from a trusted source during initialization.

\smallskip
\noindent\textbf{CPP.} Peak time is not fixed. In holidays or the days with special events, it may not be suitable to use ToU pricing model, which is generally for regular days. Therefore, CPP is developed to also handle such case, which can be implemented in a way that is similar to ToU in our system, where days become the condition of specialized piecewise function \cite{Review-2016-KMSKK,LLLLS-UDP-2014}. 

\smallskip
\noindent\textbf{RTP.} 
Real time pricing schemes allow the grid to charge customers with the nearest real time price, i.e., the price at each particular interval of time (e.g. one hour). The price can be announced one hour or a day ahead. We realize a day ahead RTP scheme proposed in paper \cite{Optimal-2010-RHG} in our system, in which the grid utility releases the predicted prices of the next 24 hours.

Let $m_h$ denote the reported power usage at hour $h$, and the pricing function which depends on three parameters $a_h$, $b_h$, $m_0 \geq 0$ be as follows:
\begin{equation}
\mathsf{RealTimePricing}(m_h)=
\begin{cases}
a_h,&\mbox{if $0 \leq m_h < m_0$,}\\
b_h,&\mbox{if $m_h \geq m_0$.}
\end{cases}
\end{equation}
In this scheme, $m_0$ is a fixed value, while $a_h$ and $b_h$ change every hour and every day. In order to allow customers to have sufficient time to schedule their electricity consumption, the GW should predict the prices of the next 24 hours (i.e.  24 $a_h$ and $b_h$) and broadcast the prices to the SMs. Let $\hat{a}[t][h]$ and $\hat{b}[t][h]$ denote the \textit{predicted} parameters for the upcoming price tariff for each hour $h$ on day $t$, the prediction model is formulated as follows:
\begin{equation}
\begin{aligned}
\hat{a}[t][h] &= k_1a[t-1][h] + k_2a[t-2][h] + k_3a[t-7][h], \\
\hat{b}[t][h] &= k_1b[t-1][h] + k_2b[t-2][h] + k_3b[t-7][h].
\end{aligned}
\end{equation}

Note that $\hat{a}[t][h]$ and $\hat{b}[t][h]$ are just the \textit{predicted} parameters. The true values ($a[t][h]$ and $b[t][h]$) will be known when the hour $h$ comes, and be used to charge customers. 
As the condition for piecewise function is power usage $x$, the access pattern of a na\"ive implementation is \textit{data-dependent}. To hide the access pattern, we use oblivious assembly functions {\small\texttt{O\_greater()}} and {\small\texttt{O\_move()}} \cite{MPML-2016-OSFMNVC,Iron-2016-FVBG}. The function \texttt{Real\_Time\_Pricing} in Fig. \ref{fig:rtp} is executed in the \textit{gateway enclave}, which avoids \texttt{if-else} branches and has no data-dependent operations. Thus it can thwart aforementioned side-channel attacks.

\begin{figure}[h]
	\centering
	\begin{minipage}[c]{0.45\textwidth}
		\centering
		\includegraphics[width=\textwidth]{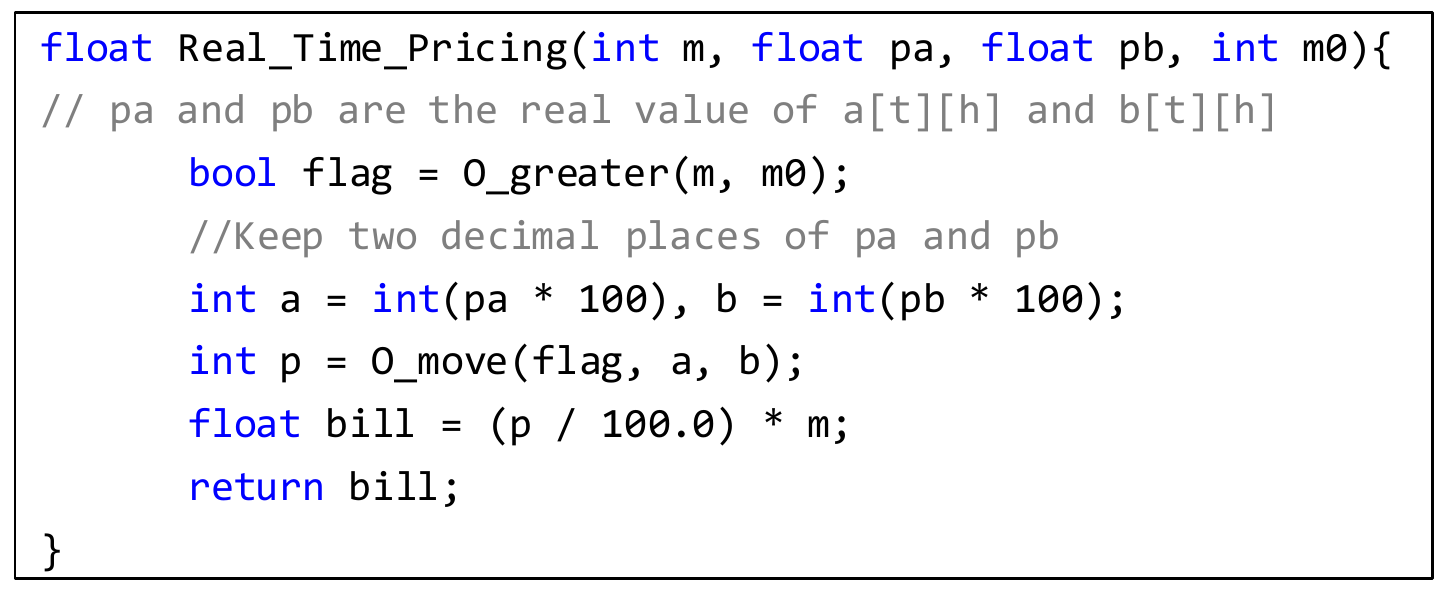}
	\end{minipage}
	\centering
	\begin{minipage}[c]{0.45\textwidth}
		\centering
		\includegraphics[width=\textwidth]{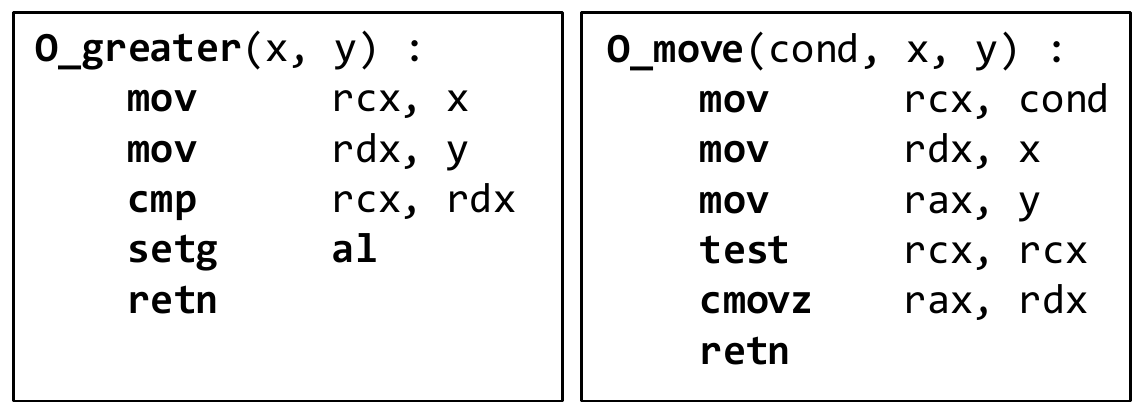}
	\end{minipage}
	\caption{Data oblivious real time pricing function.}
	\label{fig:rtp}
\end{figure}

\subsection{Load Forecasting}\label{par:s8-demand-prediction}
Current commonly used load forecasting models  include statistical based model and artificial intelligence based model \cite{LF-2015-QLHWS,Review-2016-KMSKK}. Next, we describe the secure implementation of stochastic time series method \cite{ANALYSIS-1989-IS}  and neural network based algorithm \cite{LF-2015-QLHWS}.

\smallskip
\noindent\textbf{Stochastic Time Series.} This method uses a fit function to calculate a prediction of next moment (e.g. hourly) load from previous records. An example model is \cite{ANALYSIS-1989-IS}:
\begin{equation}\label{sts}
\begin{aligned}
\mathsf{Load}(t) = &\phi_1\mathsf{Load}(t-1)+\phi_2\mathsf{Load}(t-2)\\
&+\ldots+\phi_k\mathsf{Load}(t-k)+noise(t).
\end{aligned}
\end{equation} 
To realize this model, the cryptography-based methods have to use computationally expensive homomorphic multiplication and addition. While in our system, the GW can compute the predicted load inside the \textit{gateway enclave} with decrypted data directly.


\smallskip
\noindent\textbf{Neural Network} Neural network based algorithms can figure out the relationship between referring variables and power consumption by supervised learning \cite{LF-2015-QLHWS}. Referring variables may include history consumption, day (e.g. holiday), and weather (e.g. temperature). Similarly, these data items are obtained from the \textit{control enclave}. Compared with pure cryptography schemes that need to leverage homomorphic encryption, we can run oblivious machine learning algorithms \cite{MPML-2016-OSFMNVC} on plaintext data in the \textit{gateway enclave}.

\subsection{General Functions}
Besides these three functions we described above, there are many other functions performed by the grid to improve the performance. Many of them require customers' private data as inputs. All these functions can be denoted as $f(\textbf{x}, \textbf{y})$, where $\textbf{x}$ is the privacy-related input, and $\textbf{y}$ is the other input, such as electricity price, time, and weather conditions.

To perform $f(\textbf{x}, \textbf{y})$ in a privacy-preserving way, the enclaves in our system collect $\textbf{x}$ from SMs, and request $\textbf{y}$ from a trust data source via HTTPS \cite{TownCrier-2016-ZCCJS}. For example, the grid administrator can post the latest pricing strategy on a website, where the public can easily verify and the enclaves can obtain the data they need.
From this prospective, our SecGrid system is able to support rich functionalities with privacy protection, and this feature is not available in other smart grid systems.

\section{Security Analysis}\label{par:s11-security}
Our system, \textbf{SecGrid}, aims to protect customers' privacy and execute functions securely against malicious adversaries. Specifically, our system should guarantee (1) confidentiality, (2) integrity, and (3) availability. And we will describe how the various attacks are prevented by our proposed secure protocols and oblivious operations. In addition, we analyze the security for the initialization phases, which are the foundation of our system security.

\subsection{Confidentiality}

\begin{theorem}
	Customer data from the periodical report protocol will never leak outside the enclave. Only the outputs of data aggregation, dynamic pricing, and load forecasting are revealed to the power grid company. 
\end{theorem}

\begin{proof}
	Data confidentiality comes from (a) secure communication protocols for periodic report to hide data on the fly; (b) oblivious operations in \textit{gateway enclave} to hide access pattern.
	
	\smallskip
	\noindent\textbf{(a) Outside enclave} Periodic reports from SM are encrypted with AES-GCM under key $\mathcal{K}_\mathtt{i}$ shared between SM and GW enclave. This  encryption is IND-CCA2 and any \textit{p.p.t} adversary cannot break the confidentiality without $\mathcal{K}_\mathtt{i}$. For key establishment, SM and the \textit{gateway enclave} initialize key $K_i$ using the protocol of Section \ref{SM_init}. SM generates the key and notifies the gateway in $\mathtt{Init}'$ encrypted with $\mathcal{PK}_\mathsf{gw}$. Only the \textit{gateway enclave} can decrypt $\mathtt{Init}'$ and reveal $\mathcal{K}_\mathtt{i}$ with CC's keyring. The sealing of data uses \textit{SealKey} to protect confidentiality.
	
	\smallskip
	\noindent\textbf{(b) Inside enclave} The functions are implemented on the top of oblivious operations of Section \ref{par:s7-functions}. The memory and execution patterns are no longer data-dependent.
	\begin{itemize}[leftmargin=*]
		\item \textit{Data Aggregation}. Reports are summed up according to the order of arrival. The aggregation is independent from the electricity consumption data in SM reports. 
		
		\item \textit{Dynamic Pricing}. ToU and CPP use the piecewise function but the condition is timestamp, which is not private. When RTP is used for leveled price, the condition becomes the usage $x$ which is private. We use oblivious assembly functions \cite{MPML-2016-OSFMNVC,Iron-2016-FVBG} {\small$\mathtt{O\_greater()}$} and {\small$\mathtt{O\_move()}$} to make RTP function oblivious.
		
		\item \textit{Load Forecasting}. Both regression and  neural network systems leverage oblivious learning algorithms \cite{MPML-2016-OSFMNVC}.
	\end{itemize}
	It shows the confidentiality throughout user data lifecycle.
\end{proof}

\subsection{Integrity}

\begin{theorem}
	Integrity and freshness of periodic reports from SM are achieved in the secure communication protocol. 
\end{theorem}

\begin{proof}
	AES-GCM $=(\mathcal{K},\mathcal{E},\mathcal{D})$ provides \textit{existential unforgeability}. Any \textit{p.p.t.} adversary $A$ should fail to forge a ciphertext:
	\[
	\mathsf{Adv}^{\text{exist}}_{A,(\mathcal{K},\mathcal{E},\mathcal{D})}\stackrel{\text{def}}{=}\mathsf{Pr}\left[sk\leftarrow\mathcal{K};y\leftarrow A^{\mathcal{E}_{sk}(\cdot)}:\mathcal{D}_{sk}(y)\neq\perp\right]\leq\epsilon,
	\]
	where $A$ should never receive the ciphertext $y$ in return from the encryption oracle $\mathcal{E}_{sk}(\cdot)$. The advantage $\epsilon$ is negligible.
	
	The existential unforgeability guarantees the message originates from SM. To prevent rollback attacks, we use monotonic GW counter $\mathtt{nonce}$ and SM counters $\mathtt{ctr_i}$ to prevent any steal packet of $\mathtt{ctr^{old}_i}\leq\mathtt{ctr_i}$ to forge a packet for $\mathtt{ctr_i}$. If a steal report is replayed,  \textit{gateway enclaves} can trigger an alarm.
\end{proof}

\begin{theorem}
	Integrity (and freshness) of the external database in gateway enclaves and control enclave is guaranteed. 
\end{theorem}

\begin{proof}
	\textit{Gateway enclaves} seal $K_i$ and $\mathtt{ctr_i}$ with identifier $\mathtt{ID_i}$ as the associated data in AES-GCM. The ciphertext itself is binded with $\mathtt{ID_i}$. We only need to consider rollback attacks. An attack that changes $\mathtt{ctr_i}$ to a steal value $\mathtt{ctr^{steal}_i}$ triggers an alarm to CC because a new packet from $i$-th SM with  $\mathtt{ctr_i}\textgreater\mathtt{ctr^{steal}_i}+1$ indicates the gateway has been rollbacked.
	
	The \textit{control enclave} which stores all SM initialization keys accesses the database with Merkle tree. The root hash is inside the enclave. During SM initialization, key $\mathcal{K}^\mathtt{i}_\mathsf{init}$ is labeled as ``$\mathtt{void}$'' to avoid double registration. This requires freshness which is guaranteed by Merkle tree and \textit{control enclave}. 
\end{proof}

\subsection{Availability}
We provide robustness to increase system availability by the GW restart protocol of Section \ref{GW_restart}. During short-term abrupt failures, the GW can restart the enclave and recollect the data feed from \textit{control enclave}. Lost periodic reports will be retried upon GW  request. If an adversary prevents the gateway from restarting properly, the \textit{gateway enclave} can trigger an alarm. 

\subsection{Security for Initialization}\label{security_init}
The security of our system relies on the initialization phases, including CC/GW initialization and SM initialization. The initialization of CC and GW is actually the setup of the \textit{gateway enclave} and the \textit{control enclave}. The setup phase of enclaves is secure according to our security model. Therefore, the CC/GW initialization cannot be broken by the adversary.

The UD takes an important role in the SM initialization phase, which involves four entities, the SM, the UD, the \textit{gateway enclave} and the \textit{control enclave}. As the communications between enclaves are secure and the SM cannot be compromised by the adversary, the security for the UD is essential for the initialization. Due to the limited resource of the SM, there cannot be a secure channel establish between the SM and the UD. But, the UD's role is to forward the message between the SM and the \textit{gateway enclave}, and prove the legitimacy of the \textit{gateway enclave} for the user. Therefore, even if the UD is compromised, the adversary cannot break the system, because he/she has no decryption key and cannot forge any message.

\section{Performance Evaluation}\label{par:s12-performance}
In this section, we evaluate the performance of our proposed system. We first describe the experimental environment and performance benchmarks. Then, we evaluate the networking and processing overhead of the proposed protocols. Finally, experimental results show that the functions implemented in our system represent a significant performance improvement over the existing cryptography-based solutions.

\subsection{Experimental Setup \& Performance Benchmarks}

\begin{table}[h]
	\caption{Time complexity of operations on SGX.}
	\label{benchmarks}
	\centering
	\begin{tabular}{l r}
		\hline
		\textbf{Operation} & \textbf{Time} \\
		\hline
		Create enclave & 5.6 ms\\
		Sealing (0.1 KB) & 0.015 ms\\
		Unsealing (0.1 KB) & 0.01 ms\\
		Remote Attestation & 39 ms\\
		ECDSA signing (0.1 KB) & 0.69 ms\\
		ECDSA verification (0.1 KB) & 1.21 ms\\
		AES-GCM 128-bit encryption (0.1 KB) & 0.0011 ms\\
		AES-GCM 128-bit decryption (0.1 KB) ~~~~~~~~~~~~~~~~~&0.0017 ms\\ 
		\hline
	\end{tabular}
\end{table}
\textbf{Experimental setup.} 
Our experimental platform runs Windows10 enterprise on Intel Kaby Lake i5-7500@3.40GHz processor, with 8 GB RAM and 128 GB SSD. We developed and compiled our code in Visual Studio 2015 with Intel(R) SGX SDK 1.8 and Intel(R) SGX PSW 1.8. The asymmetric cryptography used in our system for signing is 256-bit ECDSA. We use Diffie-Hellman key exchange to establish shared keys, and 128-bit AES-GCM for symmetric message encryption and authentication.

\textbf{Performance benchmarks.} 
TABLE \ref{benchmarks} provides the time complexity of the basic operations used on SGX. From this table, we can see that symmetric encryption and decryption are very fast (less than 1 microsecond). The most expensive operation is remote attestation. The main reason is that the verifier needs to interact with Intel Attestation Server via network.

\begin{figure}
	\centering
	\includegraphics[width=0.88\linewidth]{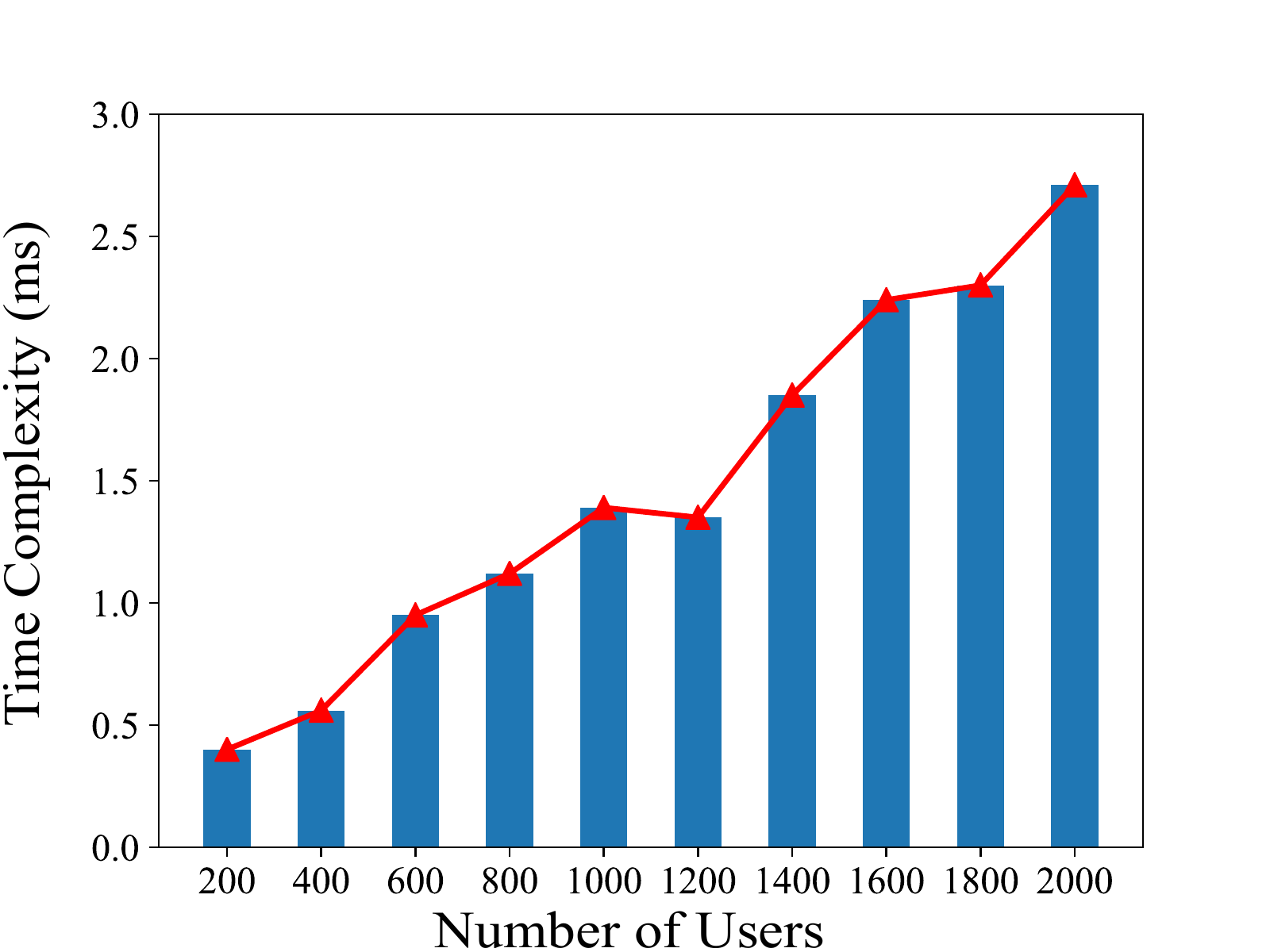}
	\centering\caption{Time complexity of data transmission.}\label{fig:transmit}
\end{figure}

\begin{figure*}[ht]
	\centering
	\begin{minipage}[h]{0.32\linewidth}
		\includegraphics[width=2.5in]{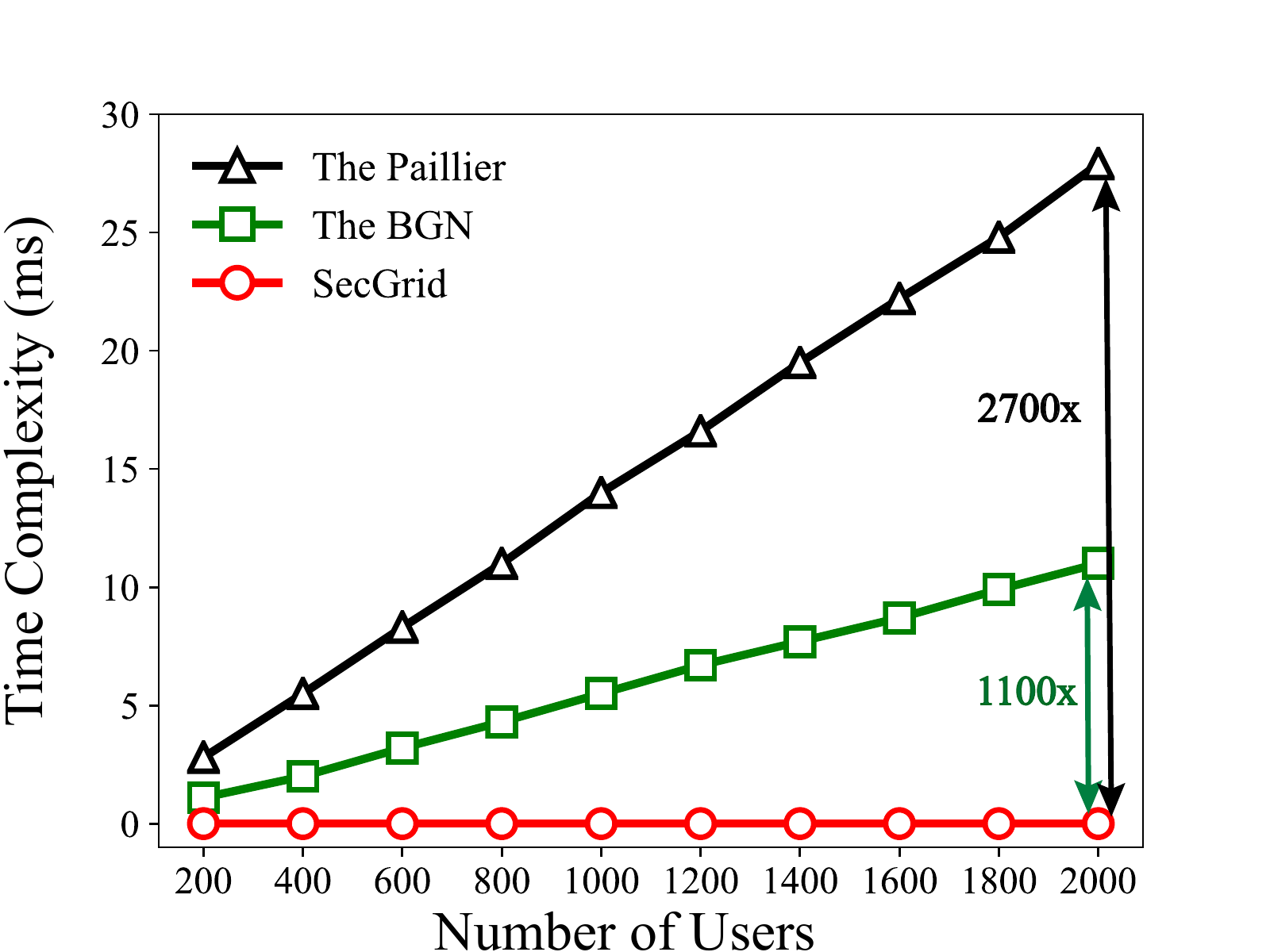}
		\centering\caption{Time complexity of data aggregation.}\label{fig:data_agg}
	\end{minipage}
	\begin{minipage}[h]{0.32\linewidth}
		\centering
		\includegraphics[width=2.5in]{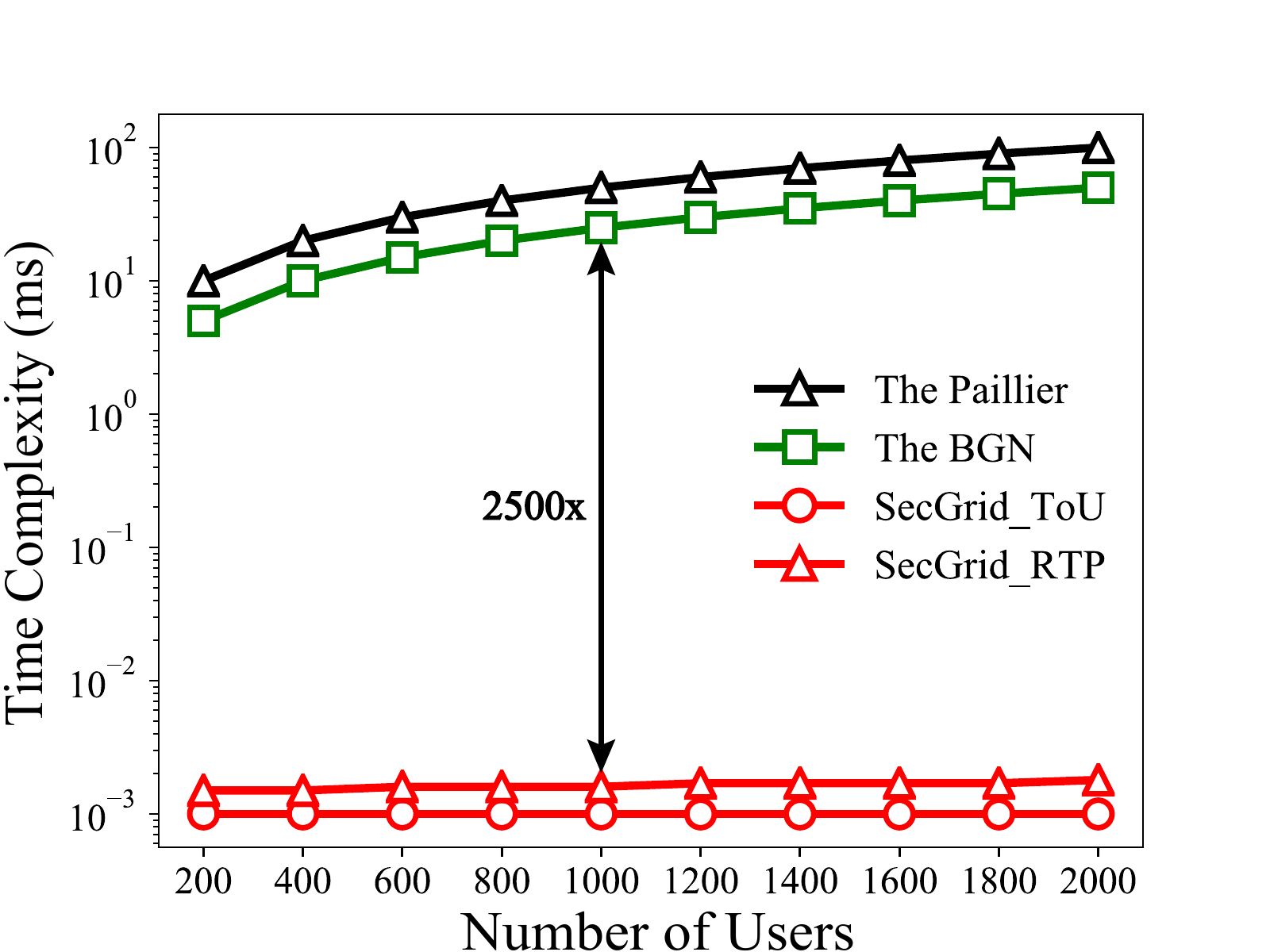}
		\caption{Time complexity of dynamic pricing.}\label{fig:dynamic_pricing}
	\end{minipage}
	\begin{minipage}[h]{0.32\linewidth}
		\centering
		\includegraphics[width=2.5in]{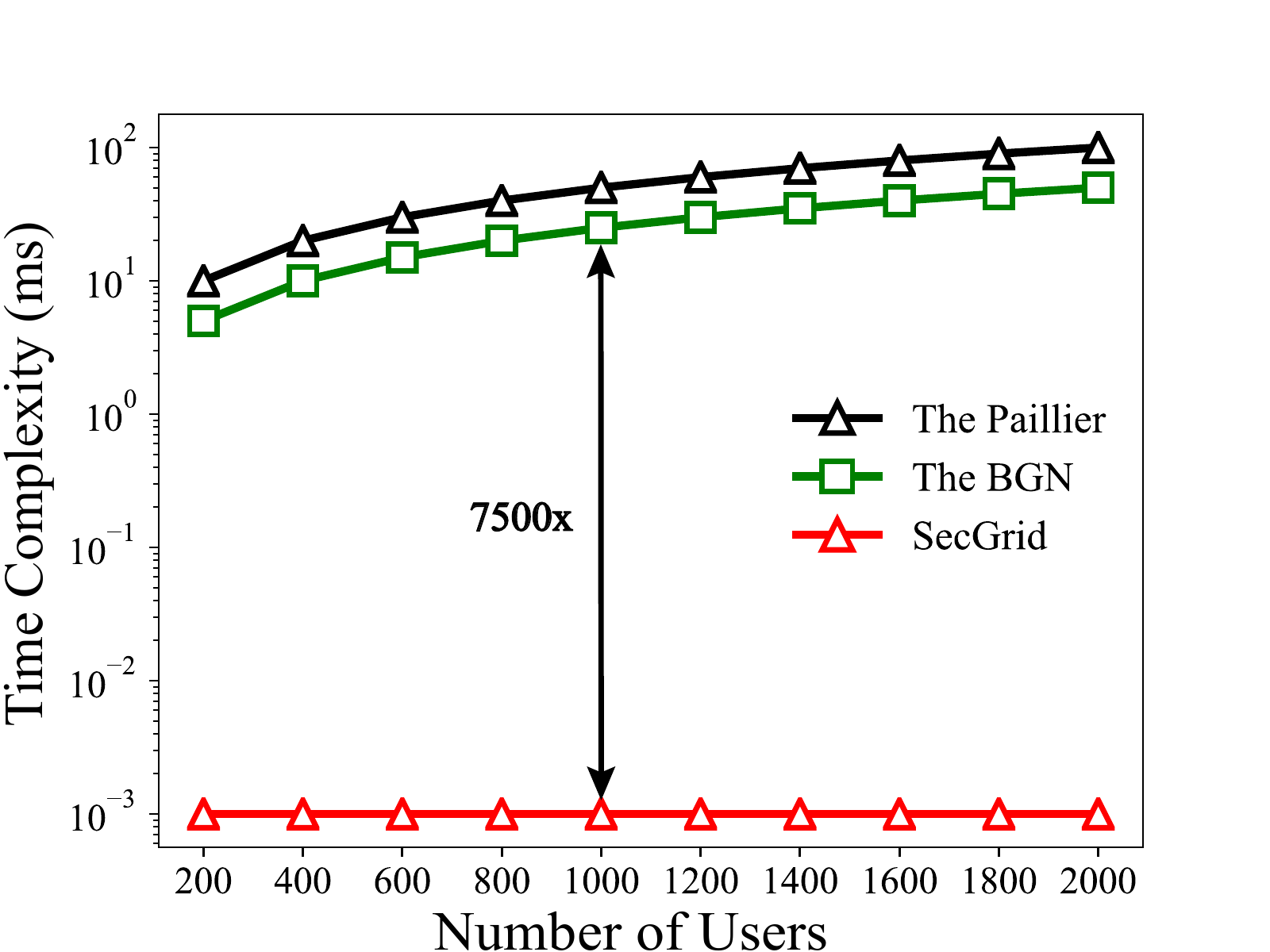}
		\caption{Time complexity of load forecasting.}\label{fig:load_forecasting}
	\end{minipage}
\end{figure*}

%
%
%

\subsection{Performance of Protocols}
The protocol overhead consists of two components: networking and processing. Since the size of network packages in these protocols is relatively small (generally, $<100$ bytes), we only consider the number of communications for the networking overhead. The processing overhead of each protocol contains partial operation as cost shown in TABLE \ref{benchmarks}. The processing and networking overhead of the proposed protocols are shown in TABLE\ref{performance_protocols}.

\begin{table}[h]
	\caption{Performance of Protocols}
	\label{performance_protocols}
	\begin{tabular}{c |c c c c c}
		\hline
		& SM & UD & \tabincell{l}{Gateway\\ Enclave \\Processing} & \tabincell{l}{Control \\ Enclave \\Processing} & Network\\
		\hline\hline
		\tabincell{c}{System\\ Initialization} & -- & -- & 41.8 ms & 40.8 ms & 4 \\
		\hline
		\tabincell{c}{SM\\ Initialization} & $\approx0$ &33.3 ms & 9.5 ms & 1.69 ms & 8 \\
		\hline
		\tabincell{c}{Periodic\\Report} & $\approx0$ & -- & 2.7 ms & 2.1 ms & 4 \\
		\hline
		\tabincell{c}{GW \\ Restart} & $\approx0$ & -- & $3.1\times n$ ms & $<1$ ms & $n+2$\\
		\hline
	\end{tabular}
\end{table}

\textbf{System initialization protocol.} Both \textit{gateway enclave} and \textit{control enclave} need remote attestation, which is the most expensive operation in this protocol (about 39 ms each). The protocol requires key exchange between the two enclaves. We can ignore the computational cost of key generation. The most expensive parts are signing and verifying, which take 0.69 ms and 1.21 ms respectively. Other operations are relatively faster, e.g. sealing, which need less than $0.1$ ms in total.

\textbf{SM initialization protocol.} This protocol involves all four entities in our system. The SM needs to generate a symmetric key and encrypt the $\mathtt{Init}$ message with it. This procedure takes less than $10\mu s$. The UD verifies the remote attestation of the \textit{gateway enclave} and encrypts the $\mathtt{Init'}$ message with $\mathcal{PK}_\mathsf{gw}$, which takes about 32 ms and 0.69 ms respectively. The \textit{gateway enclave} verifies signatures received from the UD and the \textit{control enclave}, decrypts $\mathtt{Init'}$, seals key, and  generates  a signature of $Echo$. All of these operations require only 9.5 ms. The \textit{control enclave} needs one signing and a key access procedure, which totally take about 2.1 ms.

\textbf{Periodic report protocol.} SM encrypts every report before sending it to the \textit{gateway enclave}, which, as we analyzed before, requires almost no cost. The \textit{gateway enclave} needs to decrypt and seal these reports, and encrypts a new report for the \textit{control enclave} as well as a response for the SM. These operations take about 2.1 ms. We do not include the cost of executing functions here, which will be fully evaluated in Section \ref{peformance_functions}. The \textit{control enclave} follows the same steps.

\textbf{GW restart protocol.} Once a GW restarts, the \textit{gateway enclave} needs to restart, too. It will send a request to the \textit{control enclave} to obtain fresh time, and ask each SM to send a new report to avoid losing report without sealing. This procedure requires a \textit{gateway enclave} to communicate $n+2$ times with SM and \textit{control enclave}. It takes 3.1 ms to process a report, and the \textit{gateway enclave} requires $n$ times processing.

\subsection{Performance Analysis of Functions}\label{peformance_functions}
Before the functions being executed, the encrypted data need to transmit into the \textit{gateway enclave} from outside. We test the time complexity of data transmission, the result is shown in Fig. \ref{fig:transmit}. We can see that the time complexity is very low, around 1.5 ms per 1000 users.

We implement the three functions mentioned in our paper, namely, data aggregation, dynamic pricing, and load forecasting. Existing cryptography-based schemes usually use Paillier or BGN cryptosystem to realize homomorphic operations in ciphertext. To compare our implementations with them, we also implement these  three functions with Paillier and BGN cryptosystem, which are denoted as the Paillier and the BGN in the following description. 
Details are described as follows.

\textbf{Data aggregation.} As shown in Fig. \ref{fig:data_agg}, our SecGrid has the lowest cost among all schemes. When the number of users is 2000, the time complexity of ours is 0.011 ms, which is $1100\times$ and $2700\times$ faster than the Paillier and the BGN, respectively.

\textbf{Dynamic pricing.} We implement two dynamic pricing algorithms, ToU and RTP. As for Paillier and BGN, we use fixed-price scheme, i.e., they only need to perform homomorphic multiplication, which in fact has less time complexity than existing dynamic pricing schemes. As shown in Fig. \ref{fig:dynamic_pricing}, the time complexity of our scheme is stable around $1\mu s$, which is $10^5\times$ faster than that of the Paillier and the BGN.

\textbf{Load forecasting.} Our load forecasting model is free from using Paillier and BGN that need to perform homomorphic addition and multiplication. From Fig. \ref{fig:load_forecasting}, we can see that our implementation is $2500\times$ and $7500\times$ faster than the Paillier and the BGN, respectively, and the time complexity of our protocol does not increase as the number of users increases.

\medskip
In summary, the system and SM initialization can be finished within 100 ms and 50 ms, respectively, and the processing of periodic report including three functions can also be completed within around $3\times n$ ms, where $n$ is the number of reports or users. The most expensive operation in this system is remote attestation, followed by ECDSA signing and verifying. Other operations usually take less than 10 ms. Overall, compared to cryptography-based solutions, our SecGrid system has much better performance.

\section{Conclusion}\label{par:s13-conclusion}
In this paper, we presented a practical smart grid system, named SecGrid, to enable rich functionalities without leakage of customers' private data. With our system, smart meters only need to support AES-GCM instead of heavy complex cryptography. Also, the gateway can use rich functionalities to process customers' private data at a high speed with privacy preserving. We proved that our design is sufficiently secure against malicious adversaries. We also implemented SecGrid, and thoroughly evaluate its performance. The experimental results show that our system is very efficient, as the implemented functionalities far outperform existing solutions in terms of time complexity.

%
%

\ifCLASSOPTIONcaptionsoff
  \newpage
\fi

\bibliographystyle{IEEEtran}
\bibliography{IEEEtran}
\balance
\end{document}